\documentclass[12pt]{amsart}
\usepackage{geometry}
\usepackage{graphicx}
\usepackage{amssymb}
\usepackage{epstopdf}
\usepackage{amsmath,amscd}
\usepackage{amsthm}
\usepackage{url,verbatim}
\usepackage{mathtools}

\RequirePackage[colorlinks,citecolor=blue,urlcolor=blue]{hyperref}
\usepackage{breakurl}
\theoremstyle{plain}
\newtheorem{theorem}{Theorem}

\newtheorem{lemma}[theorem]{Lemma}
\newtheorem{proposition}[theorem]{Proposition}

\newtheorem{remark}[theorem]{Remark}

\newcounter{mycount}
\newcounter{mycount2}
\newenvironment{romlist}{\begin{list}{\rm(\roman{mycount2})}%
   {\usecounter{mycount2}\labelwidth=1cm\itemsep 0pt}}{\end{list}}

\newenvironment{letlist}{\begin{list}{\rm(\alph{mycount})}%
   {\usecounter{mycount}\labelwidth=1cm\itemsep 0pt}}{\end{list}}

\numberwithin{equation}{section}
\numberwithin{theorem}{section}
\numberwithin{figure}{section}

\newcommand\HH{{\mathbb H}}
\newcommand\RR{{\mathbb R}}
\newcommand\PP{{\mathbb P}}
\newcommand\qq{\qquad}
\newcommand\q{\quad}
\newcommand\si{\sigma}
\newcommand\be{\beta}
\newcommand\al{\alpha}
\newcommand\Si{\Sigma}
\newcommand\g{\gamma}

\newcommand\FF{{\mathbb F}}

\newcommand\ta{a}
\newcommand\tb{b}
\newcommand\tc{c}

\newcommand\CC{{\mathbb C}}

\newcommand\ZZ{{\mathbb Z}}
\newcommand\VV{{\mathbb V}}
\newcommand\EE{{\mathbb E}}

\newcommand\sD{{\mathcal D}}

\newcommand\La{\Lambda}

\newcommand\eps{\epsilon}
\newcommand\ot{1-2\ }
\newcommand\resp{respectively}

\newcommand\lra{\leftrightarrow}
\newcommand\oo{\infty}

\newcommand\TT{{\mathbb T}}
\renewcommand\th{\theta}

\newcommand\Pf{\mathrm{Pf}\, }

\renewcommand\v{{\mathrm{V}}}
\newcommand\nw{{\mathrm{NW}}}
\newcommand\nea{{\mathrm{NE}}}
\newcommand\es{\varnothing}

\newcommand\Pip{\Pi^{\mathrm{poly}}}
\newcommand\Pipi{\Pief}
\renewcommand\Psi{\La}
\newcommand\Pief{\Pi_{e,f}}
\newcommand\pa{\ell}
\newcommand\Rsup{R_{\text{\rm sup}}}
\newcommand\Rsub{R_{\text{\rm sub}}}

\title[Critical surface of the hexagonal polygon model]{Critical surface of the\\ hexagonal polygon model}
\author{Geoffrey R.\ Grimmett}
\address{Statistical Laboratory, Centre for
Mathematical Sciences, Cambridge University, Wilberforce Road,
Cambridge CB3 0WB, UK}

\email{g.r.grimmett@statslab.cam.ac.uk,}
\urladdr{\url{http://www.statslab.cam.ac.uk/~grg/}}

\author{Zhongyang Li}
\address{Department of Mathematics,
University of Connecticut,
Storrs, Connecticut 06269-3009, USA}
\email{zhongyang.li@uconn.edu}
\urladdr{\url{http://www.math.uconn.edu/~zhongyang/}}

\begin{document}

\begin{abstract}
The hexagonal polygon model arises 
in a natural way via a transformation of the \ot model on the hexagonal lattice,
and it is related to the high temperature expansion of the Ising model. 
There are three types
of edge, and three corresponding parameters $\al,\be,\g>0$. 
By studying the long-range order of a certain two-edge correlation function, it is
shown that the parameter space $(0,\oo)^3$ may be divided into
\emph{subcritical} and \emph{supercritical} regions,
separated by critical surfaces satisfying an explicitly known formula.
This result complements earlier work 
on the Ising model and the \ot model. 
The proof uses the Pfaffian representation of Fisher, Kasteleyn, and Temperley
for the counts of dimers on planar graphs. 
\end{abstract}

\date{29 August 2015, revised 6 March 2016}

\maketitle

\section{Introduction}\label{sec:intro}

The polygon model studied here is a process of statistical
mechanics on the space of disjoint unions of closed loops 
on finite subsets of the hexagonal lattice $\HH$ with toroidal boundary conditions.
It arises naturally in the study of the \ot model,
and indeed the main result of the current paper is complementary
to the exact calculation of the critical surface of the \ot model reported
in \cite{GL7, GL6} (to which the reader is
referred for background and current theory of the \ot model). 
The polygon model may in addition
be viewed as an asymmetric version of the $O(n)$ model
with $n=1$ (see \cite{DPSS} for a recent reference  to the $O(n)$ model). 

Let $G=(V,E)$ be a finite subgraph of $\HH$.
The  configuration space $\Si_G$ of the polygon model is the set of all subsets $S$ of
$E$ such that every vertex in $V$ is incident to an even number of members of $S$.
The measure of the model is a three-parameter product
probability measure conditioned on belonging to $\Si_G$,
in which the three parameters are associated with the three classes of edge
(see Figure \ref{fig:hex}). 

This model may be regarded as the high temperature expansion
of a certain inhomogenous Ising model on the hexagonal lattice. The latter is a special
case of the general eight-vertex model of Lin and Wu \cite{LW90}
(see also \cite{WW89}).  Whereas
Lin and Wu concentrated on a mapping between their eight-vertex model
and a generalized Ising model, the current paper utilizes the additional symmetries
of the polymer model to calculate in closed form the equation of the critical surface
for a given choice of order parameter.  
The parameter space of the polymer model extends beyond the set of parameter values
corresponding to the classical Ising model, and thus our overall results
do not appear to follow from classical facts (see also Remarks \ref{rem:isi} and \ref{rem:sym2}).

The order parameter used in this paper is the one that corresponds to 
the two-point correlation function of the
Ising model, namely, the ratio $Z_{G,e\lra f}/Z_G$, 
where $Z_{G,e\lra f}$ is the partition function for
configurations that include a path between two edges $e$, $f$, and $Z_G$
is the usual partition function. 

Here is an overview of the methods used in this paper.
The polymer model may be transformed
into a dimer model on an associated graph, and the above ratio may be expressed 
in terms of the ratio of certain counts of dimer configurations. 
The last may be expressed (by classical results
of Kasteleyn \cite{Kast61,Kast63}, Fisher \cite{F61}, and  Temperley and Fisher \cite{TF61})
as Pfaffians of certain antisymmetric matrices.
The squares of these Pfaffians are determinants, and these
converge as $G \uparrow \HH$ to the determinants of infinite block Toeplitz
matrices. Using results of Widom \cite{HW0, HW} and others, these limits are 
analytic functions of the parameters except for certain parameter values determined by the
spectral curve of the dimer model. 
This spectral curve has an explicit representation, and this enables a computation
of the critical surface of the polygon model upon which the limiting order parameter is
non-analytic. 

More specifically, the parameter space $(0,\oo)^3$ may be partitioned into 
two regions, termed  the \emph{supercritical} and \emph{subcritical} phases.
The order parameter displays long-range order in the supercritical phase,
but not in the subcritical phase.

The results of the current paper bear resemblance to earlier results of \cite{GL6}, in which
the same authors determine the critical surface of the \ot model. The outline shape
of the main proof (of Theorem \ref{ptp})  is similar to that
of the corresponding result of \cite{GL6}. In contrast, neither result 
seems to imply the other,
and the dimer correspondence and associated calculations of the current paper
are based on a different dimer representation from that of \cite{GL6}.

The characteristics of the hexagonal lattice that are special for this work
include the properties of trivalence, planarity,  and support of a $\ZZ^2$ action. 
It may be possible to extend the results to other such graphs, such as 
the Archimedean $(3,12^2)$ lattice, and the square/octagon $(4,8^2)$ lattice.

This article is organized as follows. The polygon model is defined in Section  
\ref{sec:model}, and the main Theorem \ref{ptp} is given in Section \ref{sec:main}.
The relationship between the polygon model and the
\ot model, the Ising model,  and the dimer model
is explained in Section \ref{ssec:12}. The characteristic polynomial
of the corresponding dimer model is calculated in Section \ref{ssec:sc},
and Theorem \ref{ptp} is proved in Section \ref{sec:proof}.

\section{The polygon model}\label{sec:model}

We begin with a description of the polygon model. 
Its relationship to the \ot model is explained in  Section \ref{sec:12}.
The main result (Theorem \ref{ptp}) is given in Section \ref{sec:main}.

\subsection{Definition of the polygon model}\label{ssec:polyg}

Let the graph $G=(V,E)$ be a finite connected subgraph of the hexagonal 
lattice $\HH=(\VV,\EE)$, suitably embedded 
in $\RR^2$ as in Figure \ref{fig:hex}.
The embedding of $\HH$ is chosen in such a way
that each edge may be described by one of the three directions: 
horizontal, NW, or NE. (Later we
shall consider a finite box with toroidal boundary conditions.) 
Horizontal edges are said to be of type $\ta$, and NW edges (\resp, NE edges)
type $\tb$ (\resp, type $\tc$), as illustrated in Figure \ref{fig:hex}.
Note that $\HH$ is a bipartite graph, and we call the two classes of vertices \emph{black} and \emph{white}.

Let $\Pi$ be the product space $\Pi=\{0,1\}^E$.
The sample space of the polygon model is the subset $\Pip=\Pip(G)\subseteq\Pi$ 
containing all $\pi=(\pi_e: e \in E) \in\Pi$ such that 
\begin{equation}\label{eq:polycond}
\sum_{e\ni v} \pi_e \q\text{is either $0$ or $2$},\qq  v \in V.
\end{equation} 
Each $\pi\in\Pip$ may be considered as a union of vertex-disjoint cycles of $G$, together
with isolated vertices. We identify $\pi\in\Pi$ with the set $\{e\in E:\pi_e=1\}$  of `open' edges
under $\pi$.  Thus \eqref{eq:polycond} requires that every vertex is incident
to an even number of open edges.

\begin{figure}[htbp]
\centerline{\includegraphics*[width=0.45\hsize]{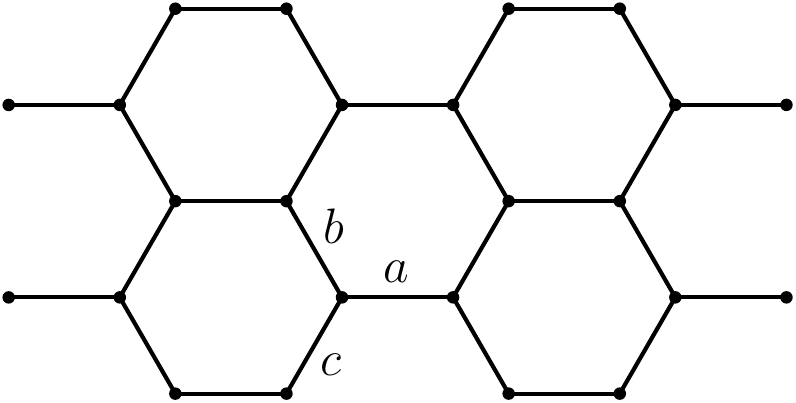}}
   \caption{An embedding of the hexagonal lattice, with the edge-types marked.}
   \label{fig:hex}
\end{figure}

Let $\eps_\ta, \eps_\tb, \eps_\tc\ne 0$. To the configuration $\pi\in\Pip$, we assign the weight
\begin{equation}\label{eq:polywt}
w(\pi)= \eps_\ta^{2|\pi(\ta)|}\eps_\tb^{2|\pi(\tb)|} \eps_\tc^{2|\pi(\tc)|},
\end{equation}
where $\pi(s)$ is  the set of open $s$-type edges of $\pi$. 
The weight function $w$ gives rise to the partition function
\begin{equation}\label{eq:pm}
Z_G(P)=\sum_{\pi\in\Pip} w(\pi).
\end{equation}
This, in turn, gives rise to a probability measure on $\Pip$ given by
\begin{equation}\label{eq:probm}
\PP_G(\pi) = \frac1{Z_G(P)}w(\pi),\qq \pi\in\Pip.
\end{equation}
The measure $\PP_G$ may be viewed as a product measure conditioned on
the outcome lying in $\Pip$.
We concentrate here on an order parameter to be given next.

It is convenient  to view the polygon model as a model on half-edges.
To this end, let $A G=(A V, A E)$ be the graph derived 
from $G=(V,E)$ by adding
a vertex at the midpoint of each edge in $E$. Let $M E=\{M e: e \in E\}$ be the set
of such midpoints, and $A V = V \cup M V$.  The edges $A E$
are precisely the half-edges of $E$, each being of the form $\langle v, Me\rangle$
for some $v \in V$ and incident edge $e \in E$.  A polygon configuration on
$G$ induces a polygon configuration on $A G$, which may be described as a 
subset of $AE$ with the property that every vertex in $AV$ has even degree. For an $a$-type edge $e\in E$,
the two half-edges of $e$ are assigned 
weight $\eps_a$ (and similarly for $b$- and $c$-type edges).
The weight function $w$ of \eqref{eq:polywt} may now be expressed as
\begin{equation}\label{eq:wt2}
w(\pi)= \eps_a^{|\pi(a)|}\eps_b^{|\pi(b)|} \eps_\tc^{|\pi(c)|},
\qq \pi\in\Pip(AG).
\end{equation}

We introduce next the order parameter of the polygon model.
Let $e, f\in ME$ be distinct midpoints of $AG$, and let $\Pipi$ be the subset of all $\pi\in\{0,1\}^{AE}$
such that: (i) every $v\in AV$ with $v\ne e,f$ is incident to an even number of open half-edges,
and (ii)  the midpoints of $e$ and $f$ are incident to exactly one open half-edge.
We define the order parameter as
\begin{equation}\label{eq:edgecor2}
M_G(e,f)=\frac{Z_{G,e\lra f}}{Z_G(P)},
\end{equation}
where 
\begin{equation}\label{eq:edgecor3}
Z_{G,e\lra f} := \sum_{\pi\in\Pipi} \eps_a^{|\pi(a)|}\eps_b^{|\pi(b)|} \eps_\tc^{|\pi(c)|}.
\end{equation}

\begin{remark}[Notation]\label{rem:pm}
We write $\eps_s$ for the parameter of an $s$-type edge, and $\eps_g$ for that
of edge $g$. For conciseness of
notation, we shall later work with the parameters
$$
\al:=\eps_a^2, \q \be:= \eps_b^2, \q \g:= \eps_\tc^2,
$$
and the main result, Theorem \ref{ptp}, is expressed in terms of these new variables.
The `squared' variables $\eps_s^2$ are introduced to permit use
of the `unsquared' signed variables $\eps_s$ in the definition 
of the polymer model on $AG$.  

The weight functions of \eqref{eq:polywt} and \eqref{eq:wt2}
are unchanged  under the sign change $\eps_s \to -\eps_s$ 
for $s\in\{a,b,c\}$.
Similarly, if the edges $e$ and $f$ have the same type, then, for $\pi\in\Pief$,
the weight $w(\pi)$ of \eqref{eq:edgecor3} is unchanged
under such sign changes. Therefore, if $e$ and $f$ have the same type,
the order parameter $M_G(e,f)$ is independent of the sign of the $\eps_g$,
and may be considered as a function of the parameters $\al$, $\be$, $\g$. 
\end{remark}

\begin{remark}[High temperature expansion]\label{rem:isi}
If  $(\al,\be,\g) \in (0,1)^3$,
the polygon model with weight function \eqref{eq:wt2}
is immediately recognized as the high temperature expansion of
an inhomogeneous Ising model on $AG$ in which the edge-interaction $J_s$ of 
an $s$-type half-edge 
satisfies  $\tanh J_s = |\eps_s|$. Under this condition,
the order parameter $M_G(e,f)$ of \eqref{eq:edgecor2} is simply
a two-point correlation function of the Ising model (see Lemma \ref{lem:12poly}).
If the $|\eps_s|$ are sufficiently small, this Ising model
is a high-temperature model, whence $M_G(e,f)$ tends to zero in the double limit
as $G \uparrow \HH_n$ and $|e-f|\to\oo$, in that order. 
It may in fact be shown (by results of \cite{ZL1,ZL-sc,lis} or otherwise)
that this Ising model has critical surface given by the equation
$\al\be+\be\g+\g\al=1$.

See \cite[p.\ 75]{Bax} and \cite{MW73,vdW} for accounts of the high temperature
expansion, and \cite{GJ09} for a recent related paper. 
The above Ising model may be viewed as a special case of the
eight-vertex model of Lin and Wu  \cite{LW90}.
It is studied further in \cite[Sect.\ 4]{GL6}.
\end{remark}

\subsection{The toroidal hexagonal lattice}\label{sec:torhex}

We will work mostly with a finite subgraph of $\HH$ subject to toroidal boundary conditions.
Let $n \ge 1$, and let $\tau_1$, $\tau_2$ be the two shifts of $\HH$, illustrated in Figure \ref{fig:hex0},
that map an elementary hexagon to the next hexagon in the given directions.
The pair $(\tau_1,\tau_2)$ generates a $\ZZ^2$ action on $\HH$, and we write 
$\HH_n=(V_n,E_n)$ for
the quotient graph of $\HH$ under the subgroup of $\ZZ^2$ generated 
by the powers $\tau_1^n$ and $\tau_2^n$.  The resulting $\HH_n$ is illustrated in Figure
\ref{fig:hex0}, and may be viewed as a finite subgraph of $\HH$ subject to toroidal
boundary conditions. We write $M_n(e,f):= M_{\HH_n}(e,f)$.

\begin{figure}[htbp]
\centering
\scalebox{1}[1]{\includegraphics*[width=0.6\hsize]{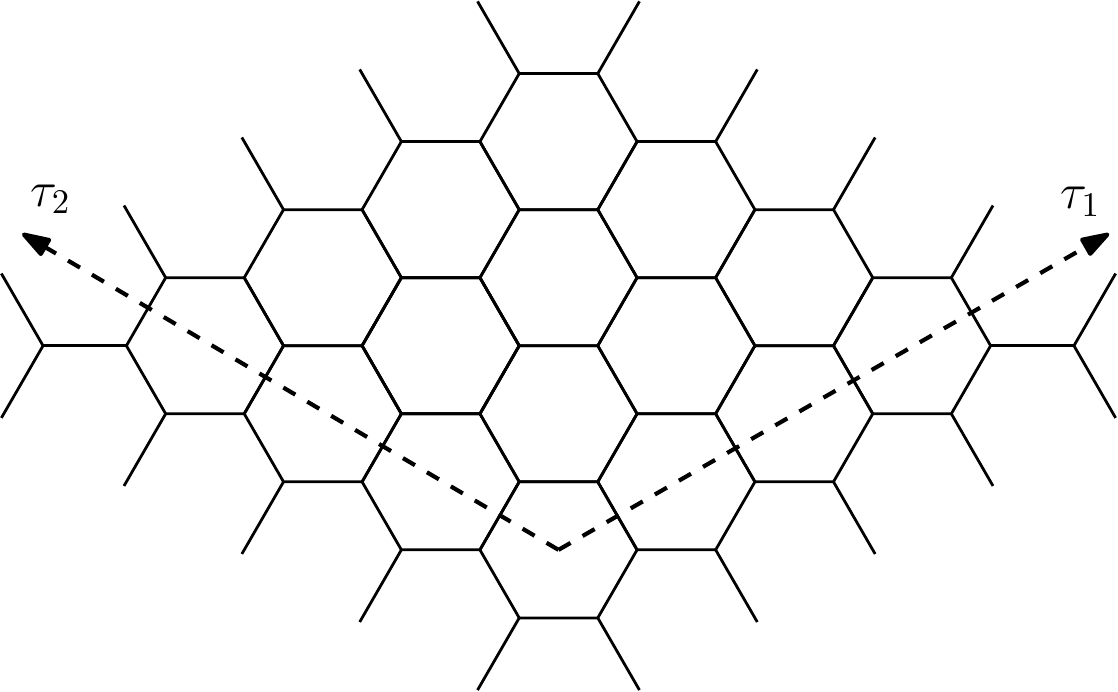}}
\caption{The graph $\HH_n$ is an $n\times n$ `diamond' wrapped onto a torus,
as illustrated here with $n=4$.}\label{fig:hex0}
\end{figure}

As in Remark \ref{rem:pm}, let 
\begin{equation}
\al=\eps_a^2,\q \be=\eps_b^2,\q \g=\eps_\tc^2.\label{sabc}
\end{equation}

\begin{theorem}\label{thm:sym}
Let $e$, $f$ be distinct edges of $\HH_n$.
The order parameter $M_n(e,f)=M_n^{\al,\be,\g}(e,f)$ is 
invariant under the change of variables $(\al,\be,\g) \mapsto (\al,\be^{-1},\g^{-1})$,
and similarly under the other two changes of variables in which exactly two of the
parameters $\al$, $\be$, $\g$ are replaced by their reciprocals. 
\end{theorem}

\begin{proof}
Let $\pi\in\Pip$ be a polygon configuration on $A\HH_n$, 
and let $\pi'$ be  obtained from $\pi$ by 
\begin{equation}\label{eq:msym}
\pi'(e) = \begin{cases} \pi(e) &\text{if  $e$ has type $a$},\\
1-\pi(e) &\text{otherwise}.
\end{cases}
\end{equation}
Since $\pi'$ is obtained from $\pi$ by adding, modulo $2$, a collection of edges
that induce an even subgraph of $\HH_n$,  we have that $\pi'\in\Pip$.
Let $w^{\al,\be,\g}(\pi)$  be the weight of $\pi$ as in \eqref{eq:wt2},
with $\al$, $\be$, $\g$ given by \eqref{sabc}. Then
\begin{equation}
w^{\al,\be,\g}(\pi)=(\be^{\# b}\g^{\# c}) w^{\al,\be^{-1},\g^{-1}}(\pi'),
\label{rl}
\end{equation}
where $\# s$ is the number of $s$-type edges in $\HH_n$.
Similarly, if $e\ne f$, then $\pi'\in\Pief$ and \eqref{rl} holds.

By \eqref{eq:msym}--\eqref{rl}, $M_n$ is 
unchanged under the map $(\al,\be,\g)\mapsto (\al,\be^{-1},\g^{-1})$.
The same proof is valid for the other two cases.
\end{proof}

\begin{remark}\label{rem:sym2}
Recall Remark \ref{rem:isi}, where it is noted that the polymer model is the high temperature 
expansion of a solvable Ising model when $(\al,\be,\g)\in(0,1)^3$. 
By Theorem \ref{thm:sym}, this results in a fairly complete picture of the behaviour of
$\lim_{n\to\oo} M_n(e,f)$ when either none or exactly two of the three parameters
lie in $(1,\oo)$. In contrast, the dimer-based methods of the current work permit
an analysis for all triples $(\al,\be,\g)\in(0,\oo)^3$, but at the price of assuming that $e$, $f$
satisfy the conditions of the forthcoming Theorem \ref{ptp}. 
\end{remark}

\subsection{Main result}\label{sec:main}

Let $e=\langle x,y\rangle$ denote the edge $e\in \EE$ of the hexagonal lattice
$\HH=(\VV,\EE)$ with endpoints $x$, $y$. 
We shall make use of a measure of distance $|e-f|$ between edges $e$ and $f$ which, 
for definiteness,  we take to be the Euclidean distance between the midpoints
of $e$ and $f$, with $\HH$ embedded in $\RR^2$ in the manner of Figure
\ref{fig:hex0} with unit edge-lengths.
Our main theorem is as follows.

\begin{theorem}\label{ptp}
Let $e,f \in \EE$ be NW edges such that: 
\begin{equation}\label{eq:condition0}
\begin{aligned}
&\text{there exists a path $\pa =\pa(e,f)$ of $A\HH_n$ from $Me$ to $Mf$,}\\
& \text{using only horizontal and NW half-edges}.
\end{aligned}
\end{equation}
Let $\eps_s \ne 0$ for $s=a,b,c$, so that $\al,\be,\g>0$, and let
\begin{equation}\label{eq:gamma}
\g_1=\left|\frac{1-\al\be}{\al+\be}\right|,
\qq \g_2=\left|\frac{1+\al\be}{\al-\be}\right|,
\end{equation}
where $\g_2$ is interpreted as $\oo$ if $\al=\be$.

\begin{letlist}
\item The limit $M(e,f)^2=\lim_{n\to\oo}M_n(e,f)^2$ exists for $\g\ne\g_1,\g_2$.

\item \emph{Supercritical case.}
Let $\Rsup$ be the set of all $(\al,\be,\g)\in(0,\oo)^3$ satisfying 
$$
 \left|\frac{1-\al\be}{\al+\be}\right|<\g<\left|\frac{1+\al\be}{\al-\be}\right|.
$$
The limit $\La(\al,\be,\g):=\lim_{|e-f|\to\oo} M(e,f)^2$ exists on
$\Rsup$, and
satisfies $\La > 0$  except possibly on some nowhere dense subset.

\item \emph{Subcritical case.} Let $\Rsub$ be the set of all $(\al,\be,\g)\in(0,\oo)^3$ satisfying 
$$
\text{either}\q \g<\left|\frac{1-\al\be}{\al+\be}\right| 
\q\text{or}\q \g> \left|\frac{1+\al\be}{\al-\be}\right|.
$$
The limit $\La(\al,\be,\g)$ exists on $\Rsub$ and satisfies $\La=0$.
\end{letlist}
\end{theorem}

The function $\La$ has a singularity when crossing between the subcritical and supercritical regions.
A brief explanation of the regions $\Rsub$ and $\Rsup$ follows. 
It turns out that the process is `critical' if and only if the spectral curve 
(see Section \ref{ssec:sc}) of the corresponding
dimer model intersects the unit torus. 
This occurs if and only the parameter-vector $(\al,\be,\g)$
is a root of the equation $\al\be+\be\g+\g\al=1$
or of any one of the three equations obtained from this equation
by the changes of variables of Theorem \ref{thm:sym}.
See Proposition \ref{prop:uvst}.  

Assumption \eqref{eq:condition0}, as illustrated in Figure \ref{fig:lef2}, is key to
the method of proof, and we present no results in the absence of this condition.
Thus, Theorem \ref{ptp} is not of itself 
a complete picture of the location of critical phenomena.
For the related 1-2 model, certain further information about the limits 
corresponding to Theorem \ref{ptp}(b,\,c) may be derived as described in
\cite{GL6} and \cite{lis}, and we do not explore that here, beyond saying that it includes
information on the rates of convergence, and on correlations unconstrained by condition \eqref{eq:condition0}.

\begin{figure}[htbp]
  \centering
\includegraphics*[scale=1]{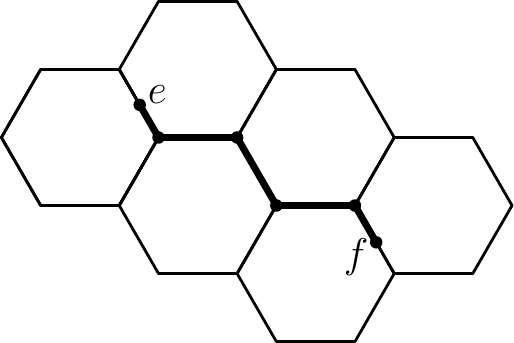}
   \caption{A path $\pa$ of NW and horizontal edges 
   connecting the midpoints of  $e$ and $f$.}\label{fig:lef2}
\end{figure}

By Remark \ref{rem:pm},
it will suffice to prove Theorem \ref{ptp} subject to
the assumption that $\eps_s>0$ for $s=a,b,c$. 

\section{The \ot and dimer models}\label{ssec:12}

We summarize next the relations between the polygon and the \ot and dimer models.

\subsection{The \ot model}\label{sec:12}

A \emph{\ot configuration} on the toroidal graph $\HH_n=(V_n,E_n)$ is
a subset $F\subseteq E_n$ such that every $v \in V_n$ is incident to either one or two members
of $F$. The subset $F$ may be expressed as a vector in the space $\Si_n=\{-1,1\}^{E_n}$
where $-1$ represents an absent edge and $1$ a present edge.
(It will be convenient later to use the space $\Si_n$ rather than the more natural $\Pi_n=\{0,1\}^{E_n}$.)
 Thus the
space of \ot configurations may be viewed as the subset of $\Si_n$ containing
all vectors $\si$ such that
$$
\sum_{e\ni v} \pi_e  \in \{1,2\}, \qq v \in V_n,
$$
where $\pi_e=\tfrac12(1+\si_e)$.

The hexagonal lattice $\HH$ is bipartite, and we colour the two
vertex-classes \emph{black} and \emph{white}.
Let $a,b,c \ge 0$ be such that $(a,b,c)\ne(0,0,0)$, 
and associate these three parameters with the edges 
as  in Figure \ref{fig:hex}.  For $\si\in\Si_n$ and $v \in V_n$, let $\si|_v$ be
the sub-configuration of $\si$ on the three edges incident to $v$,
and assign weights $w(\si|_v)$ to the $\si_v$ as in Figure \ref{fig:sign}.
We observe the states $\si_{e_{v,a}},\si_{e_{v,b}}, \si_{e_{v,c}}$,
where  $e_{v,s}$ is the edge of type $s$ incident to $v$.
The corresponding \emph{signature} is the word $\pi(e_{v,c})\pi(e_{v,b})\pi(e_{v,a})$ of length $3$.
The signature of $v$ is given as in Figure \ref{fig:sign}, together with the local weight
$w(\si|_v)$ associated with each of the six possible signatures. 

\begin{figure}[htbp]
\centerline{\includegraphics*[width=0.7\hsize]{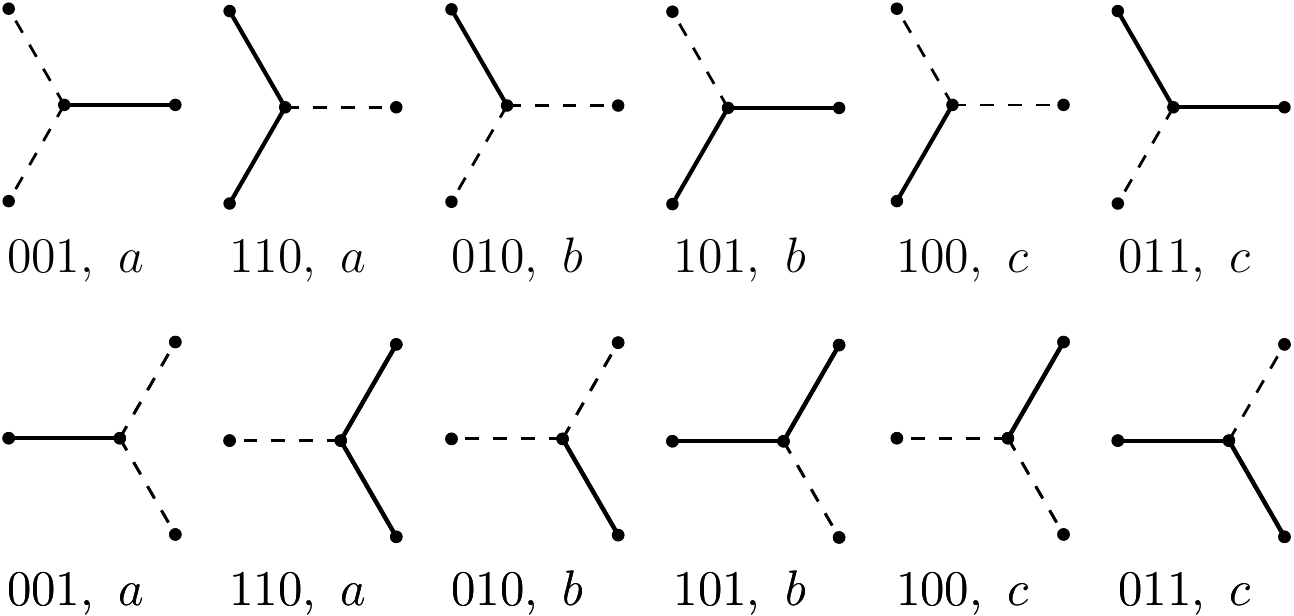}}
   \caption{The six possible local configurations $\si|_v$ at a vertex $v$  in
   the two cases of \emph{black} and \emph{white} vertices of $\HH_n$ 
   (see the upper and lower
   figures, \resp). The signature
   of each is given, and also the local weight $w(\si|_v)$ associated with each 
   local configuration.}
   \label{fig:sign}
\end{figure}

Let 
\begin{equation}\label{eq:pf-2}
w(\si) = \prod_{v\in V} w(\si|_v), \qq \si\in\Si_n,
\end{equation}
and
\begin{equation}\label{eq:pf-1}
Z_n=\sum_{\si\in\Si_n} w(\si).
\end{equation}
This gives rise to the probability measure  
\begin{equation}\label{eq:pm2}
\mu_n(\si) = \frac1{Z_n}  w(\si), \qq\si\in\Si_n.
\end{equation}
We write $\langle X\rangle_n$ for the expectation of the random
variable $X$ with respect to $\mu_n$.

The \ot model was introduced by Schwartz and Bruck \cite{SB08}
in a calculation of the capacity of a certain constrained coding system.
It has been studied by Li \cite{ZL2,ZL3}, and more recently by Grimmett and Li
\cite{GL6}. See \cite{GL7} for a review.

\subsection{The \ot model as a polygon model}\label{ssec:poly}

By \cite[Prop.\ 4.1]{GL6}, the partition function $Z_n$
of the \ot model with parameters $a$, $b$, $c$ on $\HH_n$ satisfies
$$
Z_n=\bigl(\tfrac14(a+b+c)\bigr)^{|V_n|} Z_n',
$$
where
\begin{equation}\label{pft}
Z_n'=\sum_{\si\in\Si}\prod_{v\in V_n}
\bigl(1+A\si_{v,b}\si_{v,c}+B\si_{v,a}\si_{v,c}+C\si_{v,a}\si_{v,b}\bigr),
\end{equation}
$\si_{v,s}$ denotes the state of the $s$-type edge incident to $v\in V$, and 
\begin{equation}
A=\frac{a-b-c}{a+b+c},\q 
B=\frac{b-a-c}{a+b+c},\q 
C=\frac{c-a-b}{a+b+c}.\label{abc}
\end{equation} 

Each $e =\langle u,v\rangle\in E_n$ contributes twice to the product in  \eqref{pft},
in the forms $\si_{u,s}$ and $\si_{v,s}$ for some $s\in\{a,b,c\}$. We write $\si_e$ for this common value,
and we expand \eqref{pft} to obtain a polynomial in the variables $\si_e$. In summing over $\si\in\Si_n$,
a term disappears if it contains some $\si_e$ with odd degree. Therefore,
in each monomial $M(\si)$ of the resulting polynomial, every $\si_e$ has even degree,
that is, degree either $0$ or $2$. With the monomial $M$ we associate
the set $\pi_M$ of edges $e$ for which the degree of $\si_e$ is $2$.
By examination of \eqref{pft} or otherwise, we may see that
$\pi_M$ is a polygon configuration in $\HH _n$, which is to say that
the graph $(V_n,\pi_M)$ comprises vertex-disjoint circuits (that is, closed paths
that revisit no vertex) and isolated vertices. Indeed, there is a one-to-one correspondence between
monomials $M$ and polygon configurations $\pi$. The corresponding polygon partition function 
is given at \eqref{eq:pm}
where the weights $\eps_a$, $\eps_b$, $\eps_\tc$ satisfy
\begin{equation}\label{eq:half1}
\eps_b\eps_\tc=A,\q \eps_a\eps_\tc=B, \q \eps_a\eps_b= C,
\end{equation}
which is to say that
\begin{equation}\label{eq:half2}
\eps_a^2=\frac{BC}{A},\q \eps_b^2=\frac{AC}{B},\q \eps_\tc^2=\frac{AB}{C}.
\end{equation}
Note that these squares may be negative, whence the 
corresponding $\eps_a$, $\eps_b$, $\eps_\tc$ 
are either real or purely imaginary.

The relationship between $\eps_s$ and
the parameters  $a$, $b$, $c$ is given in the following elementary lemma, the proof of which is omitted.

\begin{lemma}
Let $a\geq b\geq c>0$, and let $\eps_s$ be given by 
\eqref{abc}--\eqref{eq:half2}.
\begin{letlist}
\item Let $a<b+c$. Then  $\eps_a,\eps_b,\eps_\tc$ are purely imaginary, and moreover
\begin{romlist}
\item if $a^2<b^2+c^2$, then $0<|\eps_a|<1$, $0<|\eps_b|<1$, $0<|\eps_\tc|<1$,
\item if $a^2=b^2+c^2$, then $|\eps_a|=1$, $0<|\eps_b|<1$, $0<|\eps_\tc|<1$,
\item if $a^2>b^2+c^2$, then $|\eps_a|>1$, $0<|\eps_b|<1$, $0<|\eps_\tc|<1$.
\end{romlist}
\item If $a=b+c$, then $|\eps_a|=\infty$, $\eps_b=\eps_\tc=0$.
\item If $a>b+c$, then $\eps_a,\eps_b,\eps_\tc$ are real,
and moreover $|\eps_a|>1$, $0<|\eps_b|<1$, $0<|\eps_\tc|<1$.
\end{letlist}
\end{lemma}

Equations \eqref{eq:half1}--\eqref{eq:half2} express the $\eps_g$
in terms of $A$, $B$, $C$. Conversely, for given real $\eps_s\ne 0$,
it will be useful later to define $A$, $B$, $C$ by \eqref{eq:half1},
even when there is no corresponding \ot model.

\subsection{Two-edge correlation in the \ot model}\label{sec:eec}

Consider the \ot model on $\HH_n$ with parameters $a$, $b$, $c$, 
and specifically the two-edge correlation $\langle\si_e\si_f\rangle_n$
where $e,f\in E_n$ are distinct. 

We multiply through \eqref{pft} by $\si_e\si_f$ and expand in monomials.
This amounts to
expanding \eqref{pft} and retaining those monomials $M$
in which every $\si_g$ has even degree 
except $\si_e$ and $\si_f$, which have  degree $1$.
We may associate with $M$
a set $\pi_M$ of half-edges of $A \HH_n$ such that: (i) the midpoints
$M e$ and $M f$ have degree $1$, and (ii) every other vertex in $A V_n$ has even degree.
Such a configuration comprises a set of cycles together with a path between
$Me$ and $Mf$. The next lemma is immediate.

\begin{lemma}\label{lem:12poly}
The two-edge correlation function of the \ot model satisfies
\begin{equation}\label{eq:edgecor20}
\langle\si_e\si_f\rangle_n=\frac{Z_{n,e\lra f}}{Z_n(P)}=M_n(e,f),
\end{equation}
where the numerator $Z_{n,e\lra f}$ is given in \eqref{eq:edgecor3}, and the parameters
of the polygon model satisfy \eqref{eq:half2} and \eqref{abc}.
\end{lemma}

\subsection{The polygon model as a dimer model}\label{ssec-dimer}

We show next a one-to-one correspondence between polygon configurations on $\HH_n$ 
and dimer configurations on the corresponding \emph{Fisher graph} of $\HH_n$.
The Fisher graph $\FF_n$ is obtained from $\HH_n$ by replacing each vertex by a 
`Fisher triangle' (comprising three `triangular edges'), 
as 
illustrated in Figure \ref{fig:df}.  A \emph{dimer configuration} (or \emph{perfect matching})
is a set $D$ of edges such that each vertex is incident to exactly one edge of $D$.

Let $\pi$ be a polygon configuration on $\HH_n$ (considered as a collection of edges). 
The local configuration of $\pi$ at 
a black vertex $v \in V_n$
is one of the four configurations at the top of Figure \ref{fig:df}, 
and the corresponding local dimer configuration
is given in the lower line (a similar correspondence holds
at white vertices). 
The construction may be expressed as follows. Each edge $e$ of $\FF_n$
is either triangular or is inherited from $\HH_n$
(that is, $e$ is the central third of an edge of $\HH_n$). In the latter case, we
place a dimer on $e$ if and only if $e \notin \pi$. Having applied this rule on the edges inherited from
$\HH_n$, there is a unique allocation of dimers to the triangular edges that results in 
a dimer configuration on $\FF_n$. We write $D=D(\pi)$ for the resulting
dimer configuration, and note that the correspondence $\pi \lra D$ is one-to-one.

\begin{figure}[htbp]
\centerline{\includegraphics*{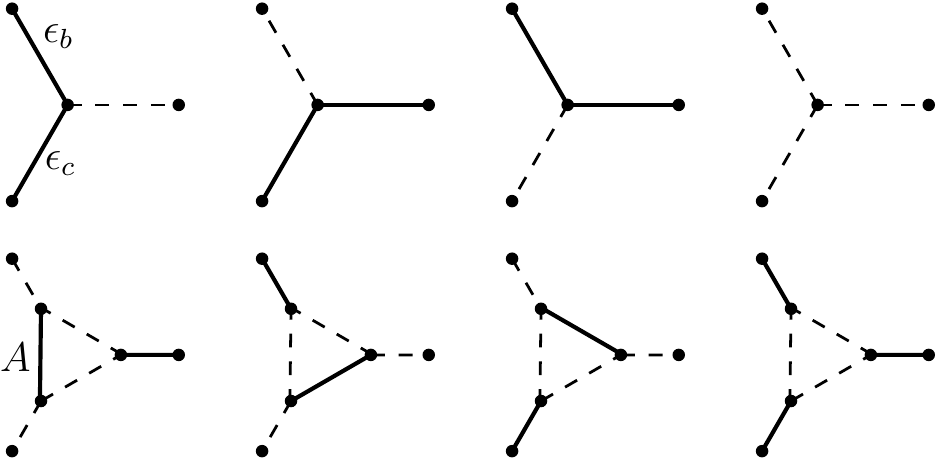}}
   \caption{To each local polygon configuration at a black vertex
   of $\HH_n$, there corresponds
   a dimer configuration on the Fisher graph $\FF_n$. The situation at 
   a white vertex is similar. In the leftmost configuration,
   the local weight of the polygon configuration  is $\eps_b\eps_\tc$,
   and in the dimer configuration $A$.}\label{fig:df}
\end{figure}

By \eqref{eq:polywt}, the weight $w(\pi)$ is the product (over $v\in V_n$) of a local weight
at $v$ belonging to the set $\{\eps_a\eps_b, \eps_b \eps_\tc, \eps_\tc\eps_a, 1\}$,
where the particular value depends on the behavior of $\pi$ at $v$ 
(see Figure \ref{fig:df} for an illustration of the four possibilities at a  black vertex).
We now assign weights to the edges of the Fisher graph $\FF_n$ in such a way that
the corresponding dimer configuration has the same weight as $\pi$.

Each edge of a Fisher triangle has one of the types:
vertical (denoted `$\v$'),  NE, 
or NW, according to its orientation. To each edge $e$ of $\FF_n$
lying in a Fisher triangle, we allocate the weight:
\begin{align*}
A\q &\text{if $e$ is vertical},\\
B\q &\text{if $e$ is NE},\\
C\q &\text{if $e$ is NW},
\end{align*}
where $A$, $B$, $C$ satisfy \eqref{eq:half1}--\eqref{eq:half2}.
The dimer partition function is given by
\begin{equation}\label{eq:pf6}
Z_n(D) := \sum_D A^{|D(\v)|} B^{|D(\nea)|} C^{|D(\nw)|},
\end{equation}
where $D(s)\subseteq D$ is the set of dimers of type $s$. 
It is immediate, by inspection of Figure \ref{fig:df}, that the correspondence $\pi \lra D$ is 
weight-preserving, and hence
$$
Z_n(D) = Z_n(P).
$$

\subsection{The spectral curve of the dimer model}\label{ssec:sc}

We turn now to the spectral curve of the weighted dimer model on $\FF_n$, 
for the background to which the reader
is referred to \cite{ZL-sc}.
The fundamental domain of $\FF_n$ is drawn in Figure \ref{11fd}, and the edges
of $\FF_n$ are oriented as in that figure. It is easily checked that this
orientation is \lq clockwise odd', in the sense that any face of $\HH_n$, 
when traversed clockwise, contains
an odd number of edges oriented in the corresponding direction.
The fundamental domain has $6$ vertices labelled $1, 2, \dots, 6$, and its weighted adjacency matrix
(or `Kasteleyn matrix') is the  $6\times 6$ matrix $W=(k_{i,j})$ with
$$
k_{i,j} = \begin{cases} w_{i,j} &\text{if $\langle i,j\rangle$ is oriented from $i$ to $j$},\\
 -w_{i,j} &\text{if $\langle i,j\rangle$ is oriented from $j$ to $i$},\\
 0 &\text{if there is no edge between $i$ and $j$},
\end{cases}
$$
where the $w_{i,j}$ are as indicated in Figure \ref{11fd}.
From $W$ we obtain a \emph{modified} adjacency (or `modified Kasteleyn')
matrix $K$ as follows.

\begin{figure}[htbp]
\centerline{\includegraphics*[scale=1]{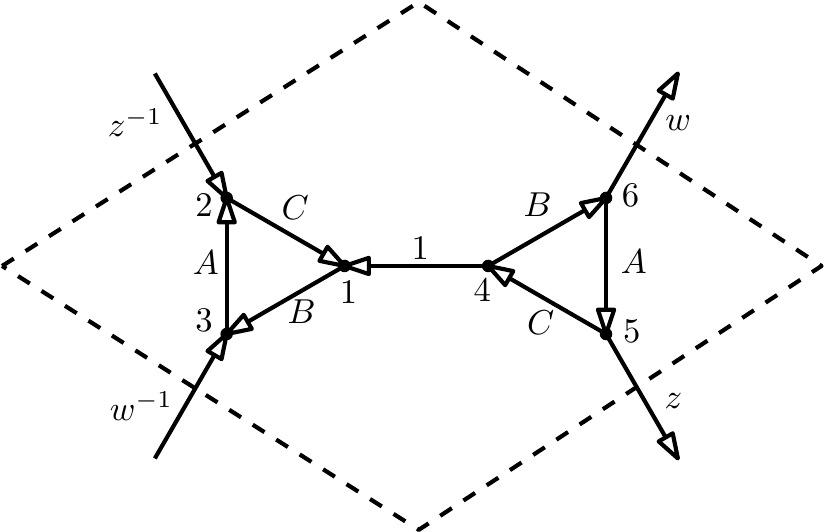}}
   \caption{Weighted $1\times1$ fundamental domain of $\FF_n$. The vertices are labelled
   $1,2,\dots,6$, and the weights $w_{i,j}$ and orientations are as indicated.
   The further weights $w^{\pm 1}$, $z^{\pm 1}$ are as indicated.}\label{11fd}
\end{figure}

We may consider the graph of Figure \ref{11fd} as being embedded in a torus, that is, 
we identify the upper left boundary and the lower right boundary, and also the 
upper right boundary and the lower left boundary, as illustrated in the figure by dashed lines. 

Let $z,w \in\CC$ be non-zero.
We orient each of the four boundaries of Figure \ref{11fd} (denoted by dashed lines) from their
lower endpoint to their upper endpoint.  The `left' and `right' of an oriented 
portion of a boundary are as viewed by a person traversing in the given direction.

Each edge $\langle u,v\rangle$ crossing a  boundary corresponds to two entries in the 
weighted adjacency matrix, indexed $(u,v)$ and $(v,u)$. If the edge starting from $u$ 
and ending at $v$ crosses an  upper-left/lower-right boundary from  left to  right 
(\resp, from right to left), we modify the adjacency matrix by multiplying the entry 
$(u,v)$ by $z$ (\resp, $z^{-1}$). If the edge starting from $u$ and ending at $v$ 
crosses an upper-right/lower-left  boundary from  left to  right (\resp, 
from  right to  left), in the modified adjacency matrix, we multiply the entry by 
$w$ (\resp, $w^{-1}$). We modify the entry $(v,u)$ in the same way.
For a definitive interpretation of Figure \ref{11fd}, the reader
is referred to the matrix following.
 
The signs of these weights are chosen to reflect the orientations of the edges.
The resulting \emph{modified adjacency matrix} (or `modified Kasteleyn matrix')  is
$$
K = \left(\begin{array}{cccccc} 0&
-C&  B&   -1& 0& 0\\C&    0&   -A&    0& -z^{-1}& 0\\
-B&  A&
0& 0& 0 &-w^{-1}\\     1&    0&    0&    0&   -C&  B\\
0& z& 0& C& 0&   -A\\     0&    0&    w&   -B&  A&
0\end{array}\right).
$$

The \emph{characteristic polynomial} is given (using Mathematica or otherwise) by
\begin{align} \label{pzw}
P(z,w)&:=\det K\\
&\phantom{:}=
1+A^4+B^4+C^4+(A^2C^2-B^2)\left(w+\frac 1w\right)\nonumber\\
&\hskip1cm +(A^2B^2-C^2)\left(z+\frac 1z\right)
+(B^2C^2-A^2)\left(\frac wz + \frac zw\right).
\nonumber
\end{align}
By \eqref{eq:half1} and \eqref{sabc},
 \begin{align*}
 P(z,w)&=1+\al^2\be^2+\al^2\g^2+
 \be^2\g^2+\al\g(\be^2-1)\left(w+\frac{1}{w}\right)
 \\
 &\hskip1cm +\al\be(\g^2-1)\left(z+\frac{1}{z}\right)
 +\be\g(\al^2-1)\left(\frac wz+\frac zw\right).
 \end{align*}

The \emph{spectral curve} is the zero locus of the characteristic polynomial, that
is, the set of roots of $P(z,w)=0$. It will be useful later
to identify the intersection of the spectral curve with the
unit torus $\TT^2=\{(z,w): |z|=|w|=1\}$.

Let
\begin{equation}\label{u}
\begin{aligned}
U&=\al\be+\be\g+\g\al-1,\\
V&=-\al\be+\be\g+\g\al+1,\\
S&=\al\be-\be\g+\g\al+1,\\
T&=\al\be+\be\g-\g\al+1.
\end{aligned}
\end{equation}

\begin{proposition}\label{prop:-1}
Let  $\eps_a,\eps_b,\eps_\tc\ne 0$, so that $\al,\be,\g>0$. 
Either the spectral curve does not intersect the unit torus $\TT^2$,  
or the intersection is a single real 
point of multiplicity $2$. 
Moreover, the spectral curve intersects $\TT^2$ 
at a single real 
point if and only if $UVST=0$, where $U,V,S,T$ are given by \eqref{u}.
\end{proposition}

\begin{proof}
The proof follows by a computation similar to those of the proofs of
\cite[Thm 9]{ZLver} and \cite[Lemma 3.2]{ZL2}.
A number of details of the proof are very close to those of
\cite{ZLver,ZL2} and are omitted. Instead, we highlight where differences arise. 

Let $\epsilon_a,\epsilon_b,\epsilon_c>0$. 
 By \eqref{sabc} and \eqref{eq:half1}--\eqref{eq:half2}, the map 
 $\psi:(A,B,C)\mapsto(\al,\be,\g)$  is a bijection between $(0,\oo)^3$ and itself. 
That $P(z,w)\geq 0$, 
 for  $(z,w)\in\TT^2$ and $(\alpha,\beta,\gamma)\in(0,\oo)^3$, 
follows by the corresponding argument 
in the proofs of \cite[Thm 9]{ZLver} and \cite[Lemma 3.2]{ZL2}.
It holds in the same way that the intersection of $P(z,w)=0$ with 
 $\TT^2$ can only be a single point of multiplicity $2$.
  
We turn now to the four points when $z,w=\pm 1$.  Note that 
 \begin{align*}
 P(1,1)&=(-1+A^2+B^2+C^2)^2=U^2,\\ 
 P(-1,-1)&=(1-A^2+B^2+C^2)^2=S^2,\\
 P(-1,1)&=(1+A^2-B^2+C^2)^2=T^2,\\
 P(1,-1)&=(1+A^2+B^2-C^2)^2=V^2,
 \end{align*}
 by \eqref{eq:half1}. Since $A,B,C\ne 0$,
 no more than one of the above four quantities can equal zero.
\end{proof}

The condition $UVST\ne 0$ may be understood as follows.
Let $\g_i$ be given by \eqref{eq:gamma}, and note
that
\begin{equation}\label{eq:gamma=}
\g_2(\al^{-1},\be)=1/\g_1(\al,\be).
\end{equation}

\begin{proposition}\label{prop:uvst}
Let $\al,\be,\g>0$ and let $U,V,S,T$ satisfy \eqref{u}.
\begin{letlist}

\item We have that $UVST=0$ if and only if $\g\in\{\g_1,\g_2\}$.

\item The region $\Rsup$ of Theorem \ref{ptp}
is an open, connected subset of $(0,\oo)^3$.

\item The region $\Rsub$ is the disjoint union of four
open, connected subsets of $(0,\oo)^3$, namely,
\begin{equation}
\begin{alignedat}{2}
\Rsub^1&=\{\g<\g_1\}\cap\{\al\be<1\}, \q 
&&\Rsub^2=\{\g<\g_1\}\cap\{\al\be>1\},\\
\Rsub^3&=\{\g>\g_2\}\cap\{\al<\be\}, \q 
&&\Rsub^4=\{\g>\g_2\}\cap\{\al>\be\}.
\end{alignedat}\label{eq:rdef}
\end{equation}
\end{letlist}
\end{proposition}

\begin{proof}
Part (a) follows by an elementary manipulation of \eqref{u}. Part (b) holds since
$\g_1<\g_2$ for all $\al,\be >0$. Part (c) is a consequence of the facts that
$\g_1=0$ when $\al+\be=0$, and $\g_2=\oo$ when $\al-\be=0$.
\end{proof}

\section{Proof of Theorem \ref{ptp}}\label{sec:proof}

\subsection{Outline of proof}
We begin with an outline of the main steps of the proof. In Section \ref{ssec:op},
the two-edge correlation of the polymer model on $\HH_n$ is expressed as the ratio of
expressions involving Pfaffians of modified adjacency matrices of
dimer models on the graph of Section \ref{ssec-dimer}. The squares of such Pfaffian ratios
are shown in Section \ref{ssec:pf} to converge to a certain determinant.
This implies the existence of the limit $M(e,f)^2$ of the two-edge correlation,
and completes the proof of part (a) of the theorem. For parts (b) and (c), one
considers the square $M(e,f)^2$ as the determinant of a block Toeplitz matrix.
By standard facts about Toeplitz matrices, the limit
$\La:=\lim_{|e-f|\to\oo} M(e,f)^2$ exists and is analytic except when the spectral
curve intersects the unit torus. The remaining claims follow by Proposition \ref{prop:-1}.

\subsection{The order parameter in terms of Pfaffians}\label{ssec:op}
By Remark \ref{rem:pm}, we shall assume without
loss of generality that  $\eps_a,\eps_b,\eps_\tc>0$. 
Let $\pa$ be the path of $A\HH_n$ connecting $Me$ and 
$Mf$ as in \eqref{eq:condition0}. To a  configuration $\pi\in\Pief$
we associate the configuration $\pi':= \pi+ \pa \in \Pip$
(with addition modulo $2$). The correspondence $\pi \lra \pi'$
is one-to-one between $\Pief$ and $\Pip$.
By considering the configurations contributing to
$Z_{n,e\lra f}$, we obtain by Lemma \ref{lem:12poly} that
\begin{equation}
M_n(e,f) = \frac{Z_{n,e\lra f}}{Z_n(P)}=
\Biggl(\prod_{g\in \pa}\eps_g\Biggr)\frac{Z_{n,\pa}(P)}{Z_n(P)},
\label{tpc}
\end{equation}
where $Z_{n,\pa}(P)$ is the partition function of polygon configurations on $A\HH_n$ 
with the weights of $s$-type half-edges along $\pa$ changed from $\eps_s$ to $\eps_s^{-1}$. 

From the Fisher graph $\FF_n$, we construct an \emph{augmented Fisher graph} $A\FF_n$ by 
placing two further vertices on each non-triangular edge of $\FF_n$, see Figure \ref{fig:AFn}.
We will construct a weight-preserving correspondence between polygon configurations on 
$A\HH _n$ and dimer configurations on $A\FF_n$. 

\begin{figure}[htbp]
  \centering
\includegraphics*{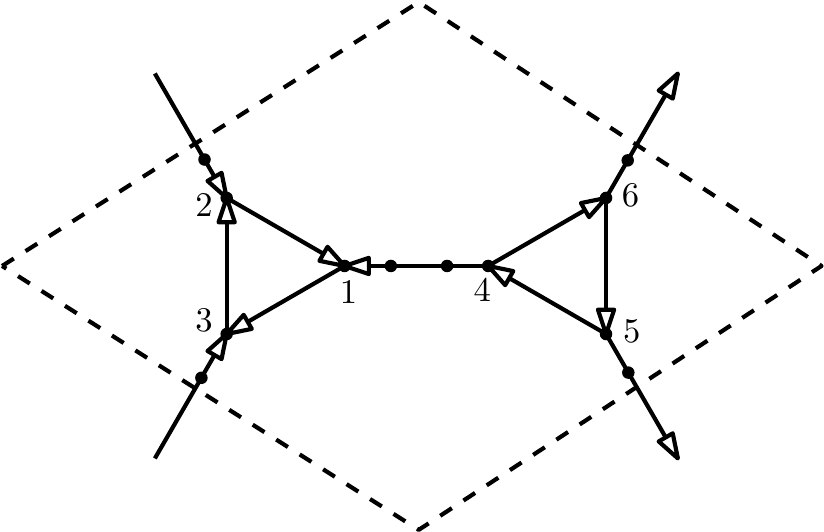}
   \caption{The fundamental domain of $A\FF_n$, which may be 
   compared with Figure \ref{11fd}.}
\label{fig:AFn}
\end{figure}

We assign weights to the edges of $A\FF_n$ as follows.
Each triangular edge of $A\FF_n$ is assigned weight $1$. 
Each non-triangular $s$-type edge of the Fisher 
graph $\FF_n$ is divided into three parts in $A\FF_n$ to which we refer
as  the left edge, the middle edge, and 
the right edge. The left edge and right edges are assigned weight $\eps_s^{-1}$, 
while the middle edge is assigned weight 1. We shall identify the characteristic polynomial
$P^A$ of this dimer model in the forthcoming Lemma \ref{lem:newdimer}.
Let $E_\pa$ be the set of left and right non-triangular edges corresponding to half-edges in 
$\pa$, and  let $V_\pa$ be the set of vertices of edges in $E_\pa$.

There is a one-to-one correspondence
between polygon configurations on $A\HH_n$ and polygon configurations on $\HH_n$.
The latter may be placed in one-to-one correspondence with dimer configurations on $A\FF_n$ as follows.
Consider a polygon configuration $\pi$ on $\HH_n$.
An edge $e\in E_n$ is present in $\pi$
if and only if the corresponding middle 
 edge of $e$ is present in the corresponding dimer configuration $D=D(\pi)$ on $A\FF_n$. 
Once the states of middle edges of $A\FF_n$ are determined, 
they generate a unique dimer configuration on $A\FF_n$.

By consideration of the particular situations that can occur
within a given fundamental domain, one obtains that the correspondence is weight-preserving
(up to a fixed factor), whence
\begin{equation*}
Z_n(P)= \Biggl(\prod_{g\in AE_n} \eps_g \Biggr)Z_n(AD),
\end{equation*}
where $Z_n(AD)$ is the partition function of the above dimer model on $A\FF_n$,
and $\eps_g$ is the parameter corresponding to an edge with the type of $g$.
A similar dimer interpretation is valid
for $Z_{n,\pa}(P)$, and thus we have
\begin{equation}\label{eq:dimrep}
\frac{Z_{n,e\lra f}}{Z_n(P)}=
\Biggl(\prod_{g\in \pa}\eps_g\Biggr)\frac{Z_{n,\pa}(P)}{Z_n(P)} = 
\left(\prod_{g\in \pa}\eps_g^{-1}\right)\frac{Z'_n(AD)}{Z_n(AD)},
\end{equation}
where $Z'_{n}(AD)$ is the partition function for dimer configurations on $A\FF_n$, 
in which an edge of $E_\pa$ has weight $\eps_g$ (where $g$ is the corresponding half-edge), and all the 
other left/right non-triangular edges have unchanged weights $\eps_g^{-1}$.

We assign a clockwise-odd orientation to the edges of $A\FF_n$ as indicated in Figure \ref{fig:AFn}.
The above dimer partition functions may be represented  in terms of the
Pfaffians of the weighted adjacency matrices corresponding to
$Z_n(AD)$ and $Z_n'(AD)$. See \cite{Kast61,Kast63,ZL1,Tes}.

Recall that $A\FF_n$ is a graph embedded in the $n\times n$ torus. Let $\g_x$ and $\g_y$ 
be two non-parallel homology generators of the torus, that is, $\g_x$ and $\g_y$ are  cycles 
winding around the torus, neither of which may be obtained from the other by continuous
movement on the torus. 
Moreover, we assume that $\g_x$ and $\g_y$ are paths in the dual graph
that meet in a unique face and that cross disjoint edge-sets.
For definiteness, we take $\g_x$ (\resp, $\g_y$) to be the upper left (\resp,  upper right) dashed cycles
of the dual triangular lattice, as illustrated in Figure \ref{tri}.
We multiply the weights of all edges crossed by $\g_x$ 
(\resp, $\g_y$) by $z$ or 
$z^{-1}$ (\resp, $w$ or $w^{-1}$), according to their orientations.
Note that $\pa$ crosses neither $\g_x$ nor $\g_y$.

\begin{figure}[htbp]
\centering
\scalebox{1}[1]{\includegraphics*[width=0.7\hsize]{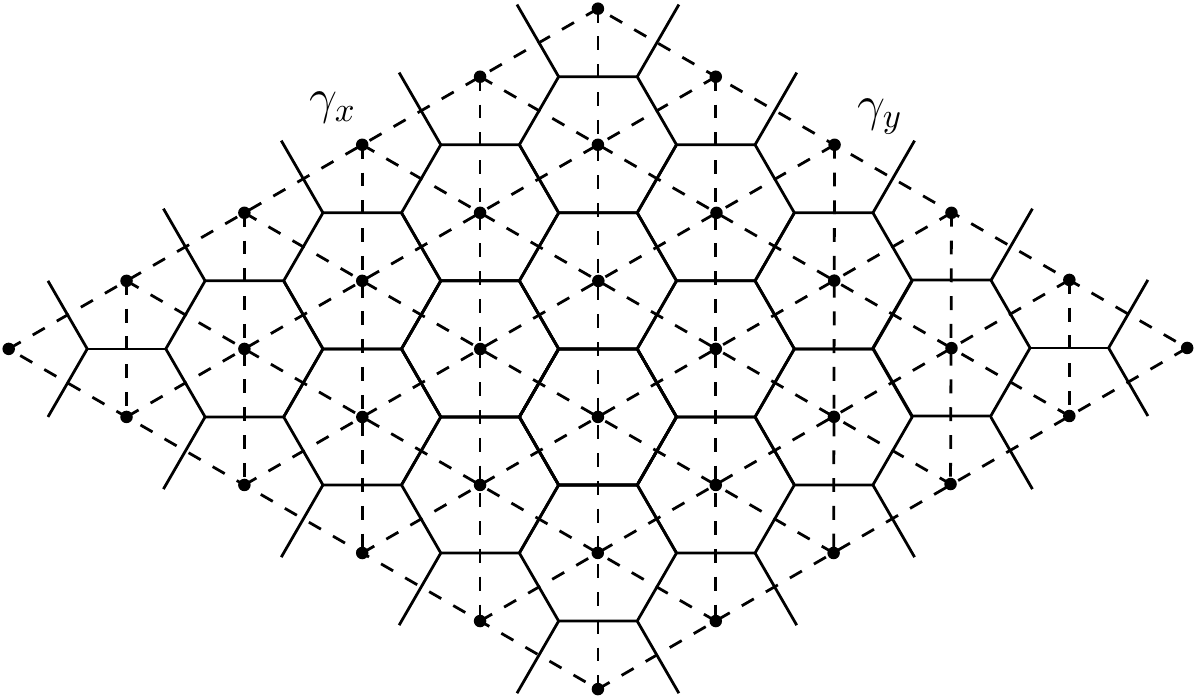}}
\caption{Two cycles $\g_x$ and $\g_y$ in the dual triangular graph of 
the toroidal graph $\HH_n$.
The upper left and lower right sides of the diamond are identified, 
and similarly for the other two sides.}\label{tri}
\end{figure}

Let $K_n(z,w)$ be the weighted adjacency matrix of
the original dimer model above, and let $K_n'(z,w)$ 
be that with the weights of $s$-type edges along $\pa$ changed 
from $\eps_s^{-1}$ to $\eps_s$. 

If $n$ is even, by \eqref{eq:dimrep}
and results of \cite{Kast61,ZL1} and \cite[Chap.\ IV]{MW73},
\begin{equation}\label{sspf}
\frac{Z_{n,e\lra f}}{Z_n(P)}=\Biggl(\prod_{g\in \pa}\eps_g^{-1}\Biggr)
\frac{-\Pf K_n'(1,1)+\Pf K_n'(-1,1)+\Pf K_n'(1,-1)+\Pf K_n'(-1,-1)}
{2Z_n(P)}, 
\end{equation}
where
\begin{equation}\label{sspf2}
2Z_n(P) = -\Pf K_n(1,1)+\Pf K_n(-1,1)+\Pf K_n(1,-1)+\Pf K_n(-1,-1).
\end{equation}
The corresponding formula when $n$ is odd is
\begin{equation*}
\frac{Z_{n,e\lra f}}{Z_n(P)}=\Biggl(\prod_{g\in \pa}\eps_g^{-1}\Biggr)
\frac{\Pf K_n'(1,1)+\Pf K_n'(-1,1)+\Pf K_n'(1,-1)-\Pf K_n'(-1,-1)}
{2Z_n(P)}, 
\end{equation*}
as explained in the discussion of `crossing orientations' of \cite[pp.\ 2192--2193]{SB08}.
The ensuing argument is essentially identical in the two cases, and therefore we may assume
without loss of generality that $n$ is even.

\subsection{The limit as $n\to\oo$}\label{ssec:pflim}

 In  studying the limit of \eqref{sspf} as $n \to\oo$, we shall require some facts about
 the asymptotic behaviour of the inverse matrix of $K_n(\th,\nu)$. We summarise these next.

The graph $A\FF_n$ may be regarded as $n\times n$ copies of
the fundamental domain 
of Figure \ref{fig:AFn}, with vertices labelled as in Figures \ref{fig:AFnlab}--\ref{fig:period}. 
We index these by $(p,q)$ with $p,q = 1,2,\dots,n$,
and let $D_{p,q}$ be the fundamental domain with index $(p,q)$.
Let $\sD=\{(p,q): D(p,q)\cap \pa\ne\es\}$, so that the cardinality of $\sD$
depends only on $|e-f|$.
Each $D_{p,q}$ contains $12$ vertices.  For $v_1,v_2\in D_{1,1}$, we
write
$K_n^{-1}(D_{p_1,q_1},v_1;D_{p_2,q_2},v_2)$ for
the $(u_1,u_2)$ entry of $K_n^{-1}$, where $u_i$ is the translate of $v_i$ lying in $D_{p_i,q_i}$.

\begin{figure}[htbp]
  \centering
\includegraphics*{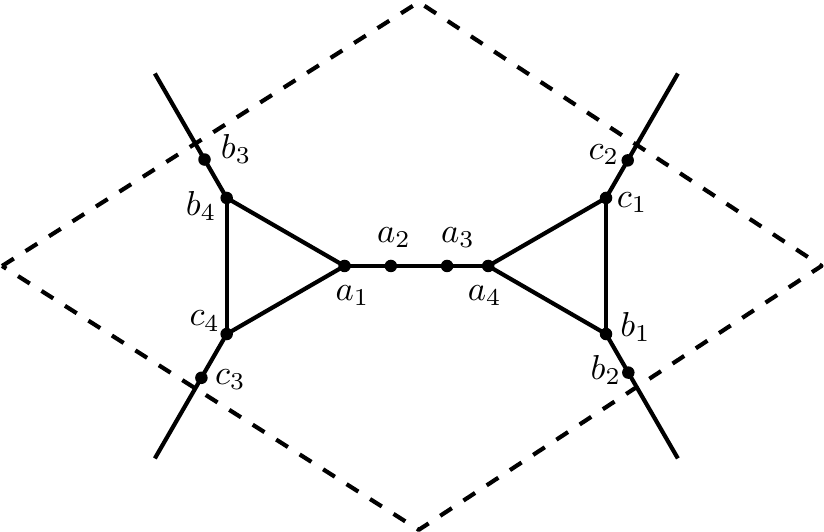}
   \caption{The fundamental domain of $A\FF_n$ with vertex-labels.}
\label{fig:AFnlab}
\end{figure}

\begin{figure}[htbp]
  \centering
\includegraphics*{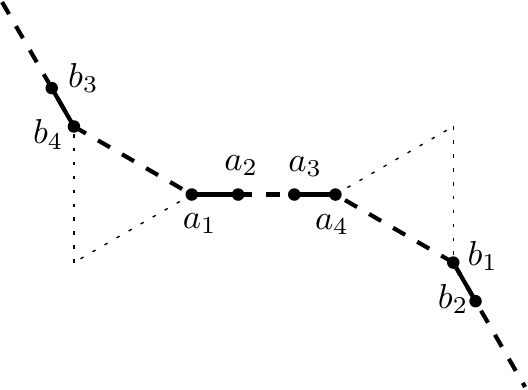}
   \caption{Part of the path $\pa$ between two NW edges.}
\label{fig:period}
\end{figure}

\begin{proposition}\label{prop:limk}
Let $\th,\nu\in\{-1,1\}$. We have that
\begin{align}
&\lim_{n\to\oo}K_n^{-1}(\th,\nu)(D_{p_1,q_1},v_{r};D_{p_2,q_2},v_{s})\label{le}\\
&\hskip2cm = -\frac{1}{4\pi^2}\int_{|z|=1}\int_{|w|=1}z^{p_2-p_1}w^{q_2-q_1}K_1^{-1}(z,w)_{v_{s},v_{r}}\,\frac{dz}{iz}\,\frac{dw}{iw},\nonumber
\end{align}
where $(p_1,q_1),(p_2,q_2)\in\sD$, and $r,s \in\{a_i,b_i:i=1,2,3,4\}$, and 
$K_1^{-1}(z,w)_{v_{s},v_{r}}$ denotes the $(v_s,v_r)$ entry of $K_1^{-1}(z,w)$. 
\end{proposition}

\begin{proof}
The limiting  entries of $K_n^{-1}(\th,\nu)$ as $n \to\oo$
can be computed explicitly using the arguments of \cite[Thm 4.3]{KOS06} and 
\cite[Sects 4.2--4.4]{RK1}, details of which are omitted here.
\end{proof}

Note that the right side of \eqref{le} does not depend on the values of $\th,\nu\in\{-1,1\}$.

\subsection{Representation of the Pfaffian ratios}\label{ssec:pf}
We return to the formulae \eqref{tpc} and \eqref{sspf}--\eqref{sspf2} for the two-edge correlation.
The matrices $K_n(\th,\nu)$ and $K_n'(\th,\nu)$ are 
antisymmetric when $\th,\nu \in\{-1,1\}$.
For $\th,\nu\in\{-1,1\}$, 
\begin{align}
\frac{\det K_n'(\th,\nu)}{\det K_n(\th,\nu)}
&=\det[K_n'(\th,\nu) K_n^{-1}(\th,\nu)] \label{kq}\\
&=\det \bigl[R_nK_n^{-1}(\th,\nu)+I\bigr],\nonumber
\end{align}
where  
\begin{equation}
R_n=K_n'(\th,\nu)-K_n(\th,\nu).\label{rd}
\end{equation}

The following argument is similar to that of \cite[Thm 4.2]{ZL1}.
Let 
\begin{equation*}
Y(\lambda)=\begin{pmatrix}
0&\lambda\\-\lambda&0 
\end{pmatrix},
\end{equation*}
and define the $4\times 4$ block matrix 
\begin{equation*}
S_s=\begin{pmatrix}
Y(\eps_s-\eps_s^{-1}) &0\\
0 &Y(\eps_s-\eps_s^{-1})
\end{pmatrix}, \qq s=a,b.
\end{equation*}

Each half-edge of $\HH _n$ along  $\pa$  corresponds to an edge of 
$A\FF_n$, namely, a left or right non-triangular edge. Moreover, the path $\pa$ has a 
periodic structure in $A\HH_n$, 
each period of which consists of four edges of $A\HH _n$, namely, 
a NW half-edge, 
followed by two horizontal half-edges, followed by another NW half-edge.
These four edges  
correspond to four non-triangular edges of $A\FF _n$ with endpoints denoted $v_{b_3}$, $v_{b_4}$, 
$v_{a_1}$, $v_{a_2}$, $v_{a_3}$, $v_{a_4}$, $v_{b_1}$, $v_{b_2}$.
See Figure \ref{fig:period}.

Let $(p,q)\in\sD$. The $12\times 12$ block
of $R_n$ with rows and columns labelled by the vertices in $D_{p,q}$ may be written as
\begin{equation}
R_n(D_{p,q},D_{p,q})=
\begin{pmatrix}
S_a&0&0\\0&-S_b&0\\0&0&0\end{pmatrix}.
\label{rpq}
\end{equation}
Each entry in \eqref{rpq} is a $4\times 4$ block, and the rows and columns are 
indexed by $v_{a_2},v_{a_1},v_{a_4},v_{a_3},v_{b_2},v_{b_1},v_{b_4},v_{b_3},v_{c_1},\dots,v_{c_4}$. 
All other entries of $R_n$ equal $0$.

Owing to the special structure of $R_n$,  
the determinant of $S_n:=R_nK_n^{-1}(\th,\nu)+I$ is the same as that
of a certain submatrix of $S_n$ given as follows. From $S_n$,
we retain all rows and columns indexed by translations (within $\sD$)
of the $v_{a_i}$ and $v_{b_j}$.
Since each fundamental domain contains four such vertices of each type,
the resulting submatrix $S_{n,\pa}$ is square with dimension $8|\sD|$. 
By following
the corresponding computations of \cite[Sect.\ 4]{ZL1} and \cite[Chap.\ VIII]{MW73}, we find that 
$\det S_n=\det S_{n,\pa}$. 

Let $X_\pa$ be the $V_\pa\times V_\pa$ block diagonal matrix with rows and columns indexed by vertices in 
$V_\pa$, and defined as follows. Adopting a suitable ordering of $V_\pa$ as above,
the diagonal $2\times 2$ blocks of $X_\pa$ are $Y(\epsilon_s-\eps_s^{-1})$, 
where $s$ depends on the type of the corresponding edge, 
and off-diagonal $2\times 2$ blocks of $X_\pa$ are 0.
Note that
\begin{equation}\label{eq:xdet}
\det X_\pa = \prod_{g\in E_\pa} \left(\eps_g-\frac1{\eps_g}\right)^2.
\end{equation}

Let $K_n^{-1}(\theta,\nu)_\pa$ 
be the submatrix of $K_n^{-1}(\theta,\nu)$ with rows and columns indexed by $V_\pa$. 
By Proposition \ref{prop:limk}, the limit
\begin{equation}\label{eq:limk}
K_\pa:= \lim_{n\to\oo}  K_n^{-1}(\theta,\nu)_\pa
\end{equation}
exists and is independent of $\th,\nu\in\{-1,1\}$.

\begin{proposition}\label{prop:zl2}
The limit $M(e,f)^2=\lim_{n\to\oo}M_n(e,f)^2$ exists and satisfies\mbox{}
\begin{equation}
M(e,f)^2=\lim_{n\to\oo}\left(\frac{Z_{n,e\lra f}}{Z_n(P)}\right)^2
=\det(X_\pa K_{\ell}^{-1}+I)\Biggl(\prod_{g\in E_{\ell}}\frac1{\epsilon_g^2}\Biggr).\label{lq0}
\end{equation}
\end{proposition}

\begin{proof}
Let $\th,\nu\in\{-1,1\}$, and assume first that $\epsilon_a,\epsilon_b\neq 1$.
By \eqref{kq}--\eqref{eq:xdet} and the discussion before the proposition,
\begin{align*}
\frac{\det K_n'(\th,\nu)}{\det K_n(\th,\nu)}
&=\det[X_\pa K_n^{-1}(\theta,\nu)_\pa+I]\\
&=\det[K_n^{-1}(\theta,\nu)_\pa+X_\pa^{-1}]\det X\\
&=\det[K_n^{-1}(\theta,\nu)_\pa+X_\pa^{-1}]
\prod_{g\in E_\pa}\left(\epsilon_g-\frac{1}{\epsilon_g}\right)^2.
\end{align*}
On taking square roots, and noting that $K_n^{-1}(\theta,\nu)_\pa+X_\pa^{-1}$
is antisymmetric,
\begin{equation*}
\frac{\Pf K_n'(\theta,\nu)}{\Pf K_n(\theta,\nu)}=
(-1)^j\Pf [K_n^{-1}(\theta,\nu)_\pa+X_\pa^{-1}]
\prod_{g\in E_\pa}\left(\epsilon_g-\frac{1}{\epsilon_g}\right),
\end{equation*}
for some $j$ that is independent of $\th$, $\nu$.

By \eqref{sspf},
\begin{align*}
2Z_{n,e\lra f}
&=(-1)^j
\Bigl\{-p(1,1) \Pf K_n(1,1)
+p(-1,1) \Pf K_n(-1,1) \\
&\hskip1.5cm  +p(1,-1)\Pf K_n(1,-1)+ p(-1,-1)\Pf K_n(-1,-1)\Bigr\}\prod_{g\in E_{\ell}}\left(1-\frac{1}{\epsilon_g^2}\right),
\end{align*}
where $p(\th,\nu)=\Pf [K_n^{-1}(\theta,\nu)_\pa+X_\pa^{-1}]$.
By \eqref{sspf2} and \eqref{eq:limk},
\begin{equation*}
\lim_{n\to\oo}\left(\frac{Z_{n,e\lra f}}{Z_n(P)}\right)^2=\left[\Pf(K_{\ell}^{-1}+X_\pa^{-1})
\prod_{g\in E_{\ell}}\left(1-\frac{1}{\epsilon_g^2}\right)\right]^2,
\end{equation*}
and \eqref{lq0} follows by \eqref{eq:xdet} and \eqref{tpc}.

Assume next that $\epsilon_a=\epsilon_b=1$. We have
\begin{equation*}
K_n'(\theta,\nu)=K_n(\theta,\nu), \qquad\mathrm{for}\ \theta,\nu =\pm 1.
\end{equation*}
Since  $X_\pa=0$ in this case, we obtain \eqref{lq0} once again.
If exactly one of $\epsilon_a$, $\epsilon_b$ equals $1$, we obtain \eqref{lq0} as above.
\end{proof}

\subsection{Proof of Theorem \ref{ptp}(a)}
As in \cite[Thm 4.3]{KOS06} and \cite[Sects 4.2--4.4]{RK1}, by Proposition \ref{prop:zl2}, 
the limit  $M(e,f)^2=\lim_{n\to\oo}M_n(e,f)^2$ 
exists and equals the 
determinant of a block Toeplitz matrix with dimension depending on $|e-f|$, and with symbol $\psi$ given by 
\begin{equation}\label{eq:Toep}
\psi(\zeta)=\frac{1}{2\pi}\int_{0}^{2\pi}
T(\zeta,\phi) \,d\phi,
\end{equation}
where  $T(\zeta,\phi)$ is the $8 \times 8$ matrix
with rows 
and columns indexed by 
$v_{a_1},v_{a_2},v_{a_3},v_{a_4}$, $v_{b_1},v_{b_2},v_{b_3},v_{b_4}$ 
(with rows and columns ordered differently)
given by 
$$
\footnotesize\begin{pmatrix}
\eps_a^{-1}+K_1^{-1}(\zeta,e^{i\phi})_{v_{a_2},v_{a_1}}\lambda_a&K_1^{-1}(\zeta,e^{i\phi})_{v_{a_2},v_{a_2}}\lambda_a&\cdots&K_1^{-1}(\zeta,e^{i\phi})_{v_{a_2},v_{b_4}}\lambda_a\\ 
-K_1^{-1}(\zeta,e^{i\phi})_{v_{a_1},v_{a_1}}\lambda_a&\eps_a^{-1}-K_1^{-1}(\zeta,e^{i\phi})_{v_{a_1},v_{a_2}}\lambda_a&\cdots&-K_1^{-1}(\zeta,e^{i\phi})_{v_{a_1},v_{b_4}}\lambda_a\\
\vdots&\vdots&\ddots&\vdots\\
K_1^{-1}(\zeta,e^{i\phi})_{v_{b_3},v_{a_1}}\lambda_b&K_1^{-1}(\zeta,e^{i\phi})_{v_{b_3},v_{a_2}}\lambda_b&\cdots&\eps_b^{-1}+K_1^{-1}(\zeta,e^{i\phi})_{v_{b_3},v_{b_4}}\lambda_b 
\end{pmatrix},
$$
and $\lambda_g=1-\eps_g^{-2}$. See \cite{HW0,HW,HW76} and the references therein for
accounts of Toeplitz matrices.

One may write
\begin{equation}
 [K_1^{-1}(z,w)]_{i,j}=\frac{Q_{i,j}(z,w)}{P^A(z,w)},\label{k1i}
\end{equation} 
where $Q_{i,j}(z,w)$ is a Laurent polynomial in $z$, $w$ derived in terms of certain cofactors of 
$K_1(z,w)$, and $P^A(z,w)=\det K_1(z,w)$ is the characteristic polynomial of the dimer model. 

\begin{lemma}\label{lem:newdimer}
The characteristic polynomial $P^A$ of the above dimer model on $A\FF_n$ satisfies 
$P^A(z,w)=(\eps_a\eps_b\eps_\tc)^{-4}P(z,w)$, where
$P(z,w)$ is the characteristic polynomial of \eqref{pzw}.
\end{lemma}

\begin{proof}
The characteristic polynomial $P^A$ satisfies $P^A(z,w)=\det K_1(z,w)$. 
Each term in the expansion of the determinant corresponds to an oriented loop configuration 
consisting of oriented cycles and doubled edges, with the property that each vertex has exactly two incident 
edges. It may be checked that there is a one-to-one correspondence 
between loop configurations on the two graphs of Figures \ref{fig:AFn} 
and \ref{11fd}, by preserving the track of each cycle and adding doubled 
edges where necessary. The weights of a pair of corresponding loop configurations 
differ by a multiplicative factor of $(ABC)^2=(\eps_a\eps_b\eps_\tc)^4$. 
\end{proof}

By the above, the limit $M(e,f)^2$ exists whenever $P^A(z,w)$ has no
zeros on the unit torus $\TT^2$. By Lemma \ref{lem:newdimer} and Proposition \ref{prop:-1},
the last occurs if and only if $UVST \ne 0$. The proof of Theorem \ref{ptp}(a) is
complete, and we turn towards parts (b) and (c).

\subsection{Proofs of Theorem \ref{ptp}(b,\,c)}\label{ssec:bc}
Consider an infinite block Toeplitz matrix $J$, viewed as the limit of an increasing 
sequence of finite truncated block Toeplitz matrices $J_n$.  
When the corresponding spectral curve does not
intersect the unit torus, 
the existence of $\det J$ as the limit of $\det J_n$ is 
proved in \cite{HW0,HW}.  By Lemma \ref{lem:newdimer} and Proposition \ref{prop:-1},
the spectral curve condition holds if and only if $UVST \ne 0$.
We deduce
the existence of the limit 
\begin{equation}\label{eq:defLa}
\La(\al,\be,\g):= \lim_{|e-f|\to\oo} M(e,f)^2,
\end{equation}
whenever $UVST \ne 0$.
By Proposition \ref{prop:-1}, the function $\Psi$ is defined on the domain 
$D:= (0,\oo)^3\setminus\{UVST=0\}$.

\begin{lemma}\label{ana} 
Assume $\al,\be,\g>0$. The function $\La$ is an analytic function of
the complex variables  $\al$, $\be$, $\g$ 
except when $UVST=0$, where $U,V,S,T$ are given by \eqref{u}.
\end{lemma}

As noted after Theorem \ref{ptp}, $\La$ is singular when $UVST=0$.

\begin{proof}
This holds  as in the proofs of \cite[Lemmas 4.4--4.7]{ZL1}. 
We consider $\Lambda$ as the determinant of a block 
Toeplitz matrix, and use Widom's formula (see \cite{HW0,HW},
and also \cite[Thm 8.7]{GL6}) to  evaluate this determinant.
As in the proof of \cite[Thm 8.7]{GL6},
$\La$ can be non-analytic only if the spectral
curve intersects the unit torus, which is to say 
(by Lemma \ref{lem:newdimer} and Proposition \ref{prop:-1})
if $UVST=0$.
\end{proof}

The equation $UVST=0$ defines a surface in the first octant 
$(0,\oo)^3$, whose complement is a union of
five open, connected components (see Proposition \ref{prop:uvst}). By Lemma \ref{ana},
$\Psi$ is analytic on each such component. 
It follows that, on any such component: either $\Psi\equiv 0$, or 
$\La$ is non-zero except possibly on a nowhere dense set.

Let $\al,\be,\g>0$. By Proposition \ref{prop:uvst}, $UVST\ne 0$ 
if and only if
\begin{equation}\label{eq:z}
\g\in(0,\g_1)\cup(\g_1,\g_2)\cup(\g_2,\oo),
\end{equation}
where the $\g_i$ are given by \eqref{eq:gamma}.

\smallskip
\noindent
\emph{Proof of Theorem \ref{ptp}(b).} 
By Proposition \ref{prop:uvst},  $UVST \ne 0$ on the open, connected region $\Rsup$.
Therefore, $\La$ is analytic on $\Rsup$. Hence,
either $\La\equiv 0$ on $\Rsup$, or $\La\not\equiv 0$ on
$\Rsup$ and the zero set $Z:=\{r=(\al,\be,\g) \in \Rsup: \La(r)=0\}$ is
nowhere dense in $\Rsup$. 
It therefore suffices to find $(\al,\be,\g)\in \Rsup$ such that $\La(\al,\be,\g)\ne 0$.

Consider the \ot model of Sections \ref{sec:12}--\ref{sec:eec} 
with $a=b>0$ and $c>4a$. By \eqref{sabc}, \eqref{abc},
and \eqref{eq:half1}, the corresponding polygon model has
parameters
$$
\al=\be=\frac{c-2a}{c+2a},\qq \g=\frac{c^2}{(c-2a)(c+2a)}.
$$
In this case, $\g_2=\oo$ and $\g\in(\g_1,\g_2)$.

By \cite[Thm 3.1]{GL6}, for almost every
such $c$, the \ot model has non-zero long-range order.
By Lemma \ref{lem:12poly},  $\La(\al,\be,\g)\ne 0$ for such $c$.

\smallskip
\noindent
\emph{Proof of Theorem \ref{ptp}(c).}
By Remark \ref{rem:isi}, when $\al,\be,\g>0$ are sufficiently small, the two-edge 
correlation function $M(e,f)$  of the polygon model equals the
two-spin correlation function $\langle \si_e\si_f\rangle$
of a ferromagnetic Ising model on $A\HH$ at high temperature. 
Since the latter has zero
long-range order, it follows that $\La=0$. Suppose, in addition, that
$\al\be<1$ and $\g<\g_1$. Since
$\La$ is analytic on $\Rsub^1$ (in the notation of \eqref{eq:rdef}),
we deduce that $\La\equiv 0$ on $\Rsub^1$.
We next extend this conclusion to $\Rsub^k$ with $k=2,3,4$.

Let $(\al,\be,\g)\in\Rsub^4$. 
By \eqref{eq:gamma=}, we have that $\al^{-1}\be<1$ and $\g^{-1}<\g_1(\al^{-1},\be)$,
so that $(\al^{-1},\be,\g^{-1})\in \Rsub^1$. 
By Theorem \ref{thm:sym},
$$
\La(\al,\be,\g) = \La(\al^{-1},\be,\g^{-1})=0.
$$
Therefore, $\La\equiv 0$ on $\Rsub^4$.

Let $(\al,\be,\g)\in\Rsub^2$,
whence $(\al^{-1},\be^{-1},\g)\in \Rsub^1$ by \eqref{eq:gamma=}. 
As above, 
$$
\La(\al,\be,\g) = \La(\al^{-1},\be^{-1},\g)=0,
$$
whence $\La\equiv 0$ on $\Rsub^2$. The case of $\Rsub^3$ can be deduced as was $\Rsub^4$.

\section*{Acknowledgements} 
This work was supported in part
by the Engineering and Physical Sciences Research Council under grant EP/I03372X/1. 
ZL acknowledges support from the Simons Foundation under  grant $\#$351813.
The authors are grateful to two referees for their suggestions,
which have improved the presentation of the work.

\bibliography{poly22}
\bibliographystyle{amsplain}

\end{document}